\documentclass{llncs}

\usepackage[utf8]{inputenc}
\usepackage[T1]{fontenc}
\usepackage{lmodern}
\usepackage{microtype}
\usepackage{fullpage}
\usepackage{subcaption}
\usepackage{amssymb, amsmath}
\usepackage{mathtools}
\usepackage{algorithm}
\usepackage[noend]{algpseudocode}
\usepackage[inline]{enumitem}
\usepackage{booktabs}
\usepackage{longtable}
\usepackage{tikz}
\usepackage{pgfplots}
\usepackage[hidelinks]{hyperref}

\usetikzlibrary{arrows, calc}
\tikzset{auto, shorten >=1pt, >=stealth}

\pgfplotsset{compat=1.14}

\newcommand {\union} {\cup}
\newcommand {\intersection} {\cap}

\newcommand {\pmap}{\ensuremath{\rightharpoonup}}

\newcommand {\Or} {\ensuremath{\vee}}
\renewcommand {\And} {\ensuremath{\wedge}}
\newcommand {\notof} {\ensuremath{\neg}}
\newcommand {\Implies} {\ensuremath{\rightarrow}}

\newcommand {\true} {\ensuremath{\mathit{true}}}
\newcommand {\false} {\ensuremath{\mathit{false}}}

\renewcommand {\models} {\ensuremath{\vDash}}

\renewcommand {\S}        {\ensuremath{\mathcal{S}}}
\newcommand {\T}        {\ensuremath{\mathcal{T}}}

\newcommand {\forcedt}[2] {\ensuremath{\mathit{forced\textrm{-}true}(#1,#2)}}
\newcommand {\forcedf}[2] {\ensuremath{\mathit{forced\textrm{-}false}(#1,#2)}}
\newcommand {\True} {\ensuremath{\mathit{True}}}
\newcommand {\False} {\ensuremath{\mathit{False}}}
\newcommand {\mymark} {\ensuremath{*}}
\newcommand {\dom} {\ensuremath{\mathit{dom}}}
\newcommand {\den}[1] {\ensuremath{[\![\texttt{#1}]\!]}}

\newcommand{\exampleend}{\hfill \tikz \draw (0, 0) -- ++(.5ex, .5ex) -- ++(0, -1ex) -- cycle;}

\begin{document}

\title{Horn-ICE Learning for Synthesizing Invariants and Contracts}
\author{Deepak D'Souza\inst{1} \and
  P.\ Ezudheen\inst{1} \and
  Pranav Garg\inst{2} \and
  P.\ Madhusudan\inst{3} \and
  Daniel Neider\inst{4} 
}

\institute{
  CSA Department, Indian Institute of Science, Bangalore, India
  \and
  Amazon India, Bangalore, India
  \and
  Department of Computer Science, University of Illinois, Urbana-Champaign, IL, USA
  \and
  Max Planck Institute for Software Systems, Kaiserslautern, Germany
}

\maketitle

\begin{abstract}
We design learning algorithms for synthesizing invariants using
Horn implication counterexamples (Horn-ICE), extending the ICE-learning model.
In particular, we describe a decision-tree learning algorithm
that learns from Horn-ICE samples, works in polynomial time, and uses
statistical heuristics to learn small trees that satisfy the samples.
Since most verification proofs can be modeled using Horn clauses, Horn-ICE
learning is a more robust technique to learn inductive annotations that
prove programs correct. 
Our experiments show that an implementation of our algorithm is able to learn 
adequate inductive invariants and contracts efficiently for a variety of sequential and concurrent
programs.
\end{abstract}


\section{Introduction}
\label{sec:intro}

Synthesizing inductive invariants, including loop invariants, pre/post 
contracts for functions, and rely-guarantee contracts for concurrent programs,
is one of the most important problems in program verification. In deductive
verification, this is often done by the verification engineer, and automating
invariant synthesis can significantly reduce the burden of building verified software,
allowing the
engineer to focus on the more complex specification and design
aspects of the code. In completely automated verification, where verification is typically
deployed for simple specifications, checking the validity of verification conditions of
recursion-free code has been mostly automated using constraint logic solving, and coupled
with invariant synthesis gives completely automated verification.

There are several techniques for finding inductive invariants, including
abstract interpretation~\cite{DBLP:conf/popl/CousotC77}, predicate abstraction~\cite{DBLP:conf/pldi/BallMMR01}, interpolation~\cite{mcmillan03,mcmillan_tacas06}, and IC3~\cite{DBLP:conf/vmcai/Bradley11}. These
techniques are typically \emph{white-box} techniques that carefully examine the program,
evaluating it symbolically or extracting unsatisfiable cores from proofs of unreachability
of error states in a bounded number of steps in order to synthesize
an inductive invariant
that can prove the program correct.

A new class of \emph{black-box} techniques based on learning has emerged in recent years
to synthesize inductive invariants~\cite{LMN14,GNMR16}. In this technique, there are two distinct agents, the Learner
and the Teacher. In each round the Learner proposes an invariant for the program, and the Teacher,
with access to a verification engine, checks whether the invariant
proves the program correct. If not, it synthesizes \emph{concrete} counterexamples that witness
why the invariant is inadequate and sends it back to the learner. The learner takes all such \emph{samples}
the teacher has given in all the rounds to synthesize the next proposal for the invariant.
The salient difference in the black-box approach is that the Learner synthesizes invariants from
concrete sample configurations of the program, and is otherwise oblivious to the program or its semantics.

It is tempting to think that the learner can learn invariants using positively and negatively labeled configurations,
similar to machine learning. However, Garg~et~al~\cite{LMN14} argued that we need a richer notion of samples for 
robust learning of inductive invariants. Let us recall this simple argument.

Consider a system with variables $\vec{x}$, with initial states captured by a predicate $\textit{Init}(\vec{x})$,
and a transition relation captured by a predicate $\textit{Trans}(\vec{x}, \vec{x}')$, and assume we want
to prove that the system does not reach a set of bad/unsafe states captured by the predicate $\textit{Bad}(\vec{x})$.
An inductive invariant $I(\vec{s})$ that proves this property needs to satisfy three constraints:
\begin{itemize}
	\item $\forall \vec{x}. \textit{Init}(\vec{x}) \Rightarrow I(\vec{x})$;
	\item $\forall \vec{x}. ~\neg(I(\vec{x}) \wedge \textit{Bad}(\vec{x}))$; and
	\item $\forall \vec{x}, \vec{x}'. I(\vec{x}) \wedge \textit{Trans}(\vec{x}, \vec{x}') \Rightarrow I(\vec{x}')$.
\end{itemize}

When a proposed invariant fails to satisfy the first two conditions, the verification engine can indeed come
up with configurations labeled positive and negative to indicate ways to correct the invariant. However, when
the third property above fails, it cannot come up with a single configuration labeled positive/negative; and
the most natural counterexample is a \emph{pair} of configurations $c$ and $c'$, with the instruction to the
learner that if $I(c)$ holds, then $I(c')$ must also hold. These are called \emph{implication counterexamples}
and the ICE (Implication Counter-Example) learning framework developed by Garg~et~al.\ is a robust learning
framework for synthesizing invariants~\cite{LMN14}. 
Garg~et~al.\ define several learning algorithms for learning invariant synthesis, in particular learning algorithms based on decision trees that can learn Boolean combinations of Boolean predicates and inqualities that compare numerical predicates to arbitrary thresholds~\cite{GNMR16}. 

Despite the argument above, it turns out that implication counterexamples are \emph{not sufficient} for learning
invariants in program verification settings. This is because reasoning in program verification is more stylized,
to deal \emph{compositionally} with the program. In particular, programs with function calls and/or
concurrency are not amenable to the above form of reasoning. In fact, it turns out that most reasoning
in program verification can be expressed in terms of \emph{Horn clauses}, where the Horn clauses have some formulas
that need to be synthesized.

For example, consider the imperative program snippet:
\[ I_\textit{pre}(\vec{x},y)~~ S(\text{mod}~\vec{x});~ y:=\textit{foo}(\vec{x}); ~ I_\textit{post}(\vec{x},y) \]
that we want to show correct, where $S$ is some straight-line program that modifies $\vec{x}$,
$I_\textit{pre}$ and $I_\textit{post}$ are some annotation (like the contract of a function we are synthesizing.
Assume that we are synthesizing the contract for $\textit{foo}$ as well,  
and assume the post-condition for $\textit{foo}$ is $PostFoo(res, \vec{x})$,
where $res$ denotes the result it returns.
Then the verification condition that we want to be valid is
\[ \left( I_\textit{pre}(\vec{x}, y) \wedge \textit{Trans}_S(\vec{x}, \vec{x}') \wedge \textit{PostFoo}(y', \vec{x}') \right) \Rightarrow I_\textit{post}(\vec{x}', y'), \]
where $\textit{Trans}_S$ captures the logical semantics of the snippet $S$ in terms
of how it affects the post-state of $\vec{x}$.

In the above, all three of the predicates $I_\textit{pre}$, $I_\textit{post}$ and $\textit{PostFoo}$ need to be synthesized.
When a learner proposes concrete predicates for these, the verifier checking the above logical
formula may find it to be invalid, and find concrete valuations $v_{\vec{x}}, v_y, v_{\vec{x}'}, v_{y'}$ 
for $\vec{x}, y, \vec{x}', y'$ that
makes the above implication false. However, notice that the above cannot be formulated as
a simple implication constraint. The most natural constraint to return to the learner is
\[ \left( I_\textit{pre}(v_{\vec{x}}, v_{\vec{v_y}})  \wedge \textit{PostFoo}(v_{y'}, v_{\vec{x}'}) \right) \Rightarrow I_\textit{post}(v_{\vec{x}'}, v_{y'}), \]
asking the learner to meet this requirement when coming up with predicates
in the future. The above is best seen as a \emph{Horn} Implication CounterExample (Horn-ICE).

The primary goal of this paper is to build Horn-ICE (Horn implication counterexample) learners for
learning predicates that facilitate inductive invariant and contract synthesis for proving safety
properties of programs. It has been observed in the literature that most program verification mechanisms
can be stated in terms of proof rules that resemble Horn clauses~\cite{DBLP:conf/pldi/GrebenshchikovLPR12}; in fact, the formalism 
of \emph{constrained Horn clauses} has emerged as a robust general mechanism for capturing program
verification problems in logic~\cite{DBLP:conf/cav/GurfinkelKKN15}. Consequently, whenever a Horn clause fails, it results in a Horn-ICE
sample that can be communicated to the learner, making Horn-ICE learners a much more general
mechanism for synthesizing invariants and contracts.

Our main technical contribution is to devise a decision-tree based Horn-ICE algorithm.
Given a set of (Boolean) predicates over configurations of programs and numerical functions 
that map configurations to integers, the goal of the learning algorithm is to synthesize 
predicates that are \emph{arbitrary} Boolean combinations of the Boolean predicates
and atomic predicates of the form $n \leq c$, where $n$ denotes a numerical function,
where $c$ is arbitrary. The classical decision-tree learning algorithm learns such
predicates from samples labeled $+/-$ only~\cite{Quinlan86}, and the work by Garg~et~al extends
decision-tree learning to learning from ICE-samples~\cite{GNMR16}. In this work, we extend
the latter algorithm to one that learns from Horn-ICE samples.

Extending decision-tree learning to handle Horn samples turns out to be non-trivial.
When a decision tree algorithm reaches a node that it decides to make a leaf and label it TRUE,
in the ICE-learning setting it can simply \emph{propagate} the constraints across the implication
constraints. However, it turns out that for Horn constraints,
this is much harder. Assume there is a single invariant we are synthesizing and
we have a Horn sample $(s_1 \wedge s_2) \Rightarrow s'$ and we decide
to label $s'$ false when building the decision tree. Then we must later turn \emph{at
least one of $s_1$ and $s_2$ to false}. This choice makes the algorithms and propagation
much more complex, and ensuring that the decision tree algorithm will always construct
a correct decision tree (if one exists) and work in polynomial time becomes much harder. 
Furthermore, statistical measures based on entropy for choosing attributes (to split each node) get more complicated 
as we have to decide on a more complex logical space of Horn constraints between samples.

The contributions of this paper are the following:
\begin{enumerate}
	\item A robust decision-tree learning algorithm that learns using Horn implication counterexamples, runs in polynomial time (in the number of samples) and has a bias towards learning smaller trees (expressions) using statistical measures for choosing attributes. The algorithm also guarantees that a decision-tree
          consistent with all samples is created, provided there exists one.
	\item We show that we can use our learning algorithm to learn over an \emph{infinite} countable set of predicates ${\mathcal P}$, and we can ensure learning is complete (i.e., that will find an invariant if one is expressible using the predicates ${\mathcal P}$).
	\item An implementation of our algorithm, extending the classical IC3 decision-tree algorithm, and an automated verification tool built with our algorithm for synthesizing invariants. We evaluate our algorithm for finding loop invariants and summaries for sequential programs and also Rely-Guarantee contracts in concurrent programs.
\end{enumerate}

The paper is structured as follows. In Section~\ref{sec:overview} we present an overview of Horn ICE invariant synthesis; in Section~\ref{sec:learner}, we describe the decision tree based algorithm for learning invariant formulas from Horn ICE samples; in Section~\ref{sec:horn-solver}, we describe the algorithm that propagates the data point classifications across Horn constraints; we describe the node/attribute selection strategies used in the decision tree based learning algorithm in Section~\ref{sec:attribute-selection} and the experimental evaluation in Section~\ref{sec:experiments}.


\subsection*{Related Work}
Invariant synthesis is the central problem in automated program verification and, over the years, several techniques have been proposed for synthesizing invariants, including abstract interpretation~\cite{DBLP:conf/popl/CousotC77}, interpolation~\cite{mcmillan03,mcmillan_tacas06}, IC3 and PDR~\cite{DBLP:conf/vmcai/Bradley11,DBLP:conf/cav/KarbyshevBIRS15}, predicate abstraction~\cite{DBLP:conf/pldi/BallMMR01}, abductive inference~\cite{DBLP:conf/oopsla/DilligDLM13}, as well as synthesis algorithms that rely on constraint solving~\cite{DBLP:conf/pldi/gulwani08,invgen,DBLP:conf/cav/ColonSS03}. 
Subsequent to~\cite{DBLP:conf/pldi/GrebenshchikovLPR12}, there has been a lot of work towards Horn-clause solving~\cite{BjornerMR13,DBLP:conf/cav/BeyenePR13}, using a combination of these techniques.
Complementing these techniques are data driven invariant synthesis techniques, the first ones to be proposed being Daikon~\cite{DBLP:conf/icse/ErnstCGN00} that learns \emph{likely} program invariants and Houdini~\cite{DBLP:conf/fm/FlanaganL01} that learns conjunctive inductive invariants.
Data-driven invariant synthesis has seen renewed interest, lately~\cite{DBLP:conf/cav/SharmaNA12,DBLP:conf/esop/0001GHALN13,DBLP:conf/sas/0001GHAN13,DBLP:conf/cav/0001LMN13,LMN14,DBLP:conf/cav/0001A14,Zhu:2015:LRT:2784731.2784766,DBLP:conf/kbse/PavlinovicLS16,DBLP:conf/pldi/PadhiSM16,DBLP:conf/icse/NguyenKWF12,GNMR16,Zhu:2016:ALS,DBLP:conf/sas/BrockschmidtCKK17,FKB17}.
When the program manipulates complex data-structures, arrays, pointers, etc., or when one needs to reason over a complicated memory model and its semantics, 
the invariant for the correctness of the program might still be simple. In such a scenario, a black-box, data-driven \emph{guess and check} approach, guided by a finite set of program configurations, has been shown to be advantageous.
However, implication counter-examples proposed by Garg et al.~\cite{LMN14} are not sufficient for learning invariants in general program verification settings. 
Grebenshchikov et al.~\cite{DBLP:conf/pldi/GrebenshchikovLPR12} have shown that most reasoning in program verification is expressed in terms of Horn clauses.
Subsequent to~\cite{DBLP:conf/pldi/GrebenshchikovLPR12}, there has been a lot of work towards Horn-clause solving~\cite{BjornerMR13,DBLP:conf/cav/BeyenePR13}. 
SeaHorn~\cite{DBLP:conf/cav/GurfinkelKKN15} is a verification framework that translates verification conditions of a program to horn clauses that can be solved using several backend solvers.
In the context of data-driven invariant synthesis, 
our work generalizes the ICE learning model~\cite{LMN14} to Horn counter-examples, and we build a decision-tree based Horn-ICE learner for learning invariant annotations in this model.


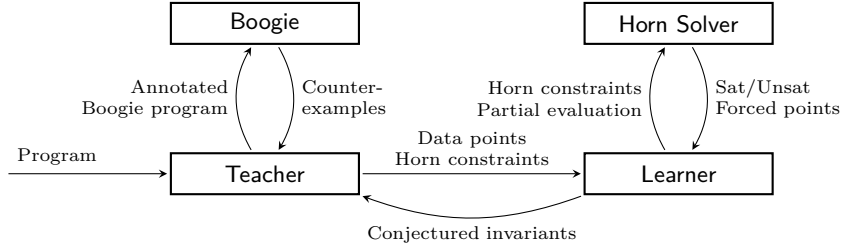
\begin{figure}[t!h]
	\centering
	
	\begin{tikzpicture}
		\begin{scope}[font=\sffamily]
			\node[draw, minimum width=25mm, minimum height=4ex, thick] (teacher) at (0, 0) {Teacher};
			\node[draw, minimum width=25mm, minimum height=4ex, thick] (boogie) at (0, 2) {Boogie};
			\node[draw, minimum width=25mm, minimum height=4ex, thick] (learner) at (5.5, 0) {Learner};
			\node[draw, minimum width=25mm, minimum height=4ex, thick] (solver) at (5.5, 2) {Horn Solver};
		\end{scope}
		
		\draw[->, shorten <=1pt] (teacher) edge[bend left] node[text width=20mm, font=\scriptsize, align=right] {Annotated \\ Boogie program} (boogie);
		\draw[->, shorten <=1pt] (boogie) edge[bend left]  node[text width=20mm, font=\scriptsize, align=left] {Counter- \\ examples} (teacher);
		\draw[->, shorten <=1pt] (learner) edge[bend left] node[text width=30mm, font=\scriptsize, align=right] {Horn constraints \\ Partial evaluation} (solver);
		\draw[->, shorten <=1pt] (solver) edge[bend left]  node[text width=20mm, font=\scriptsize, align=left] {Sat/Unsat \\ Forced points} (learner);
		\draw[->, shorten <=1pt] (teacher) edge[]  node[text width=30mm, font=\scriptsize, align=center] {Data points \\ Horn constraints} (learner);
		\draw[shorten <=1pt] (learner.south west) edge[->, bend left=20]  node[text width=50mm, font=\scriptsize, align=center] {Conjectured invariants} (teacher.south east);
		
		\node[font=\scriptsize, anchor=south] (program) at ([xshift=-15mm] teacher.west) {Program};
		\draw (program.south west) edge[->] (teacher.west);
	\end{tikzpicture}
	
  	\caption{Architecture of the Horn-ICE invariant synthesis tool} \label{fig:horn-ice-arch}
\end{figure}

\section{Overview}
\label{sec:overview}

We give an overview of the Horn-ICE invariant synthesis framework in this section.
Fig.~\ref{fig:horn-ice-arch} shows the main components of our Horn-ICE invariant synthesis framework.
The Teacher is given a program along with a specification, and based on the style of proof the teacher wants to carry out, she determines the
kind of invariants needed (a name for each invariant, and the set of
variables or terms it may mention) and the corresponding verification
conditions they must satisfy.
The Learner conjectures a concrete invariant for each invariant name, 
and communicates these to the Teacher.
The Teacher plugs in these conjectured invariants in the program and
asks a verification engine (in this case Boogie, but we could use any suitable program verifier) to check if the conjectured
invariants suffice to prove the specification.
If not, Boogie will return a counter-example showing why the
conjectured invariants do not constitute a valid proof.
The Teacher passes these counter-examples to the Learner.
The Learner now learns a new invariants that are consistent
with the set of counter-examples given by the Teacher so far.
The Learner frequently invokes the Horn Solver and Horn propagation engine to guide it
in the process of building concrete invariants that are consistent
with the set of counter-examples given by the Teacher. The Teacher and
Learner go through a number of such rounds, until the Teacher
finds that the concrete invariants supplied by the Learner constitute
a correct proof. 

We illustrate the working of the tool with an example.
Fig.~\ref{fig:prog-text} shows a concurrent program adapted from
\cite{Mine14}, with two threads $T_1$ and $T_2$ that access two shared
variables \texttt{x} and \texttt{y}.
The precondition says that initially $x = y = 0$.
The postcondition asserts that when the two threads terminate, the
state satisfies $x \leq y$.
Let us say the Teacher is interested in proving this specification using
Rely-Guarantee style reasoning \cite{Jones83b,XuRH97}.
In this proof technique we need to come up with invariants associated
with the program points $P0$--$P4$ in $T_1$ and $Q0$--$Q4$ in $T_2$ (as in
the Floyd-Hoare style proofs of sequential programs), as well as
two ``two-state'' invariants $G1$ and $G2$, which act as the
``guarantee'' on the interferences caused by each thread respectively.
Symmetrically, $G2$ and $G1$ can be thought of as the conditions that
the threads $T_1$ and $T_2$ can ``rely'' on, respectively, about the 
interferences from the other thread.
Fig.~\ref{fig:prog-chc} shows a partial list of the verification
conditions that the named invariants need to satisfy, in order to
constitute a valid Rely-Guarantee proof of the program.
The VCs are grouped into four categories: ``Adequacy'' and
``Inductiveness'' are similar to the requirements for sequential
programs, while ``Guarantee'' requires that the atomic statements of
each thread satisfy its promised guarantee, and ``Stability'' requires
that the invariants at each point are stable under interferences from
the other thread.
We use the notation ``$\den{x:=x+1}$'' to denote the two-state (or
``before-after'') predicate describing the semantics of the statement
``\texttt{x:=x+1}'', which in this case is the predicate 
$x' = x + 1 \And y' = y$. The guarantee invariants $G1$ and $G2$ are
similar predicates over the variables $x,y,x',y'$, describing the
possible state changes that an atomic statement in a thread can
effect.
The invariants $P0$--$P4$ and $Q0$--$Q4$ are
predicates over the
variables $x,y$. We use the notation $P0'$ to denote the predicate
$P0$ applied to the variables $x',y'$.

\begin{figure}[th]
	\centering

	\begin{subfigure}[b]{.325\linewidth}
		\begin{scriptsize}
		\begin{verbatim}
         Pre: x = y = 0                  
                                           
    T1         ||          T2              
                                           
P0  while (*) {        Q0  while (*) {     
P1    if (x < y)       Q1    if (y < 10)   
P2      x := x + 1;    Q2      y := y + 3  
P3  }                  Q3  }               
P4                     Q4                      

           Post: x <= y                    
		\end{verbatim}
	\end{scriptsize}
	\caption{The program} \label{fig:prog-text}
	\end{subfigure}%
\hskip 2cm
\begin{subfigure}[b]{.475\linewidth}
  \centering

  \begin{scriptsize}

  \begin{tabular}{|ll|ll|} \hline 
    & Adequacy & & Inductiveness \\ \hline
  1.& $(x = 0 \And y = 0) \Implies P0$ & 1. & $P0 \Implies P1 \And P4$ \\
  2.& $P4 \And Q4 \Implies (x \leq y)$ & 2. & $P1 \And (x < y) \Implies P2$ \\
    &                              & 3. & $P2 \And \den{x := x + 1} \Implies P3'$ \\
    &                              & 4. & $P3 \Implies P0$ \\ 
    &                              &    & $\cdots$ \\ \hline
    & Stability  & & Guarantee \\ \hline
  1.& $P0 \And G2 \Implies P0'$ & 1.& $P2 \And \den{x := x + 1} \Implies G1$  \\ 
  2.& $P1 \And G2 \Implies P1'$ & 2.& $Q2 \And \den{y := y + 3} \Implies G2$ \\ 
    & $\cdots$                  &   & \\ \hline
  \end{tabular}
  \end{scriptsize}
	\caption{The verification conditions}\label{fig:prog-chc} 
\end{subfigure}

\caption{A concurrent program and the corresponding verification conditions for a Rely-Guarantee proof}\label{fig:1}
\end{figure}

The Teacher asks the Learner to synthesize the invariants $P0$--$P4$ and
$Q0$--$Q4$ over the variables $x,y$, and $G1$ and $G2$ over the
variables $x,y,x',y'$.
As a first cut the Learner conjectures ``$\true$'' for all these
invariants.
The Teacher encodes the VCs in Fig.~\ref{fig:prog-chc} as annotated
procedures in Boogie's programming language, plugs in $\true$ for
each invariant, and asks Boogie if the annotated program verifies.
Boogie comes back saying that the ensures clause corresponding to VC
Adequacy-1 may fail, and gives a counter-example say $\langle x
\mapsto 2, y \mapsto 1 \rangle$ which satisfies $P4$ and $Q4$, but
does not satisfy $x \leq y$.
The Teacher conveys this counter-example as a Horn sample
$d_1 \And d_2 \Implies \false$ to the Learner, where $d_1$ is the data
point $\langle P4,2,1\rangle$ and $d_2$ is the data
point $\langle Q4,2,1\rangle$.
We use the convention that the data points are vectors in which the
first component is the value of a 
``location'' variable ``$l$'' which takes one of the values ``$P0$'',
``$P1$'', etc, while the second and third components are values of $x$
and $y$ respectively.
This Horn constraint is represented graphically in the bottom of
Fig.~\ref{fig:prog-hc}.

To focus on the technique used by the Learner, which is the heart of
this work, let us pan to several rounds ahead, where the Learner has
accumulated a set of counter-examples given by the Teacher, as shown
in Fig.~\ref{fig:prog-hc}.
The Learner's goal is simply to find a small (small in expression size)
 invariant $\varphi$ (from a
finite class of 
formulas comprising Boolean combinations of some base predicates), that
is \emph{consistent} with the given set of Horn constraints. By
``consistent'', we mean that for each Horn constraint of the form $d_1
\And \cdots \And d_k \Implies d$, whenever each of
$d_1, \ldots, d_k$ satisfy $\varphi$, it is the case that $d$ also
satisfies $\varphi$.

Our Learner uses a decision tree based learning technique. Here the
internal nodes of the decision tree are labelled by the base
predicates (or ``attributes'') and leaf-nodes are classified as
``True'', ``False'', or ``?'' (for ``Unclassified'').
Each leaf node in a decision tree represents a logical formula which
is the conjunction of the node labels along the
path from the root to the leaf node, and the whole tree represents the
disjunction of the formulas corresponding to the leaf nodes labelled
``True''.
The Learner builds a decision tree for a given set of Horn constraints
incrementally, starting from the root node.
Each leaf node in the tree has a corresponding subset of the data-points
associated with it, namely the set of points that satisfy the formula
associated with that node.
In each step the Learner can choose to mark a node as ``True'', or
``False'', or to split a node
with a chosen attribute and create two child nodes associated with
it.

Before marking a node as ``True'' or ``False'' the Learner would like
to make sure that this ``preserves'' the consistency of the decision tree
with respect to the set of Horn constraints.
For this he calls the Horn Solver/Propagation component, which reports
whether the proposed extension of the partial valuation is
indeed consistent with the given set of Horn constraints, and if so
which are the data-points which are ``forced'' to be true or false.
For example, let us say the Learner has constructed the partial
decision tree shown in Fig.~\ref{fig:prog-dec-tree}, where node $n_4$
has already been set to ``True'' and nodes $n_2$ and $n_5$ are
unclassified.
He now asks the Horn Solver if it is okay for him to turn node $n_2$
``True'', to which the Horn Solver replies ``Yes'' since this extended
valuation would still be consistent with the set of Horn constraints
in Fig.~\ref{fig:prog-hc}.
The Horn Solver also tells him that the extension would force the
data-points $d_{12}$, $d_8$, $d_7$, $d_5$, $d_4$, $d_3$, $d_1$ to be
$\true$, and the point $d_2$ to $\false$.

The Learner uses this information to go ahead and set $n_2$ to
``True'', and also to make note of the fact that $n_5$ is now a ``mixed''
node with some points that are forced to be $\true$ (like $d_1$) and
some $\false$ (like $d_2$).
Based on this information, the Learner may choose to split node $n_5$ 
next.
After completing the decision tree, the Learner may send the
conjecture in which $P0$--$P4$, $G1$ and $G2$ are set to $\true$, and
$Q1$--$Q4$ are set to $\false$.
The Teacher sends back another counter-example to this conjecture, and
the exchanges continue for several rounds.
Finally, our Learner eventually makes a successful conjecture like: $x
\leq y$ for $P0$, $P1$, $P3$, $P4$, and $Q0$--$Q4$; $x < y$ for $P2$;
$y = y' \And x' \leq y'$ for $G1$; and $x = x' \And y \leq y'$ for $G2$.

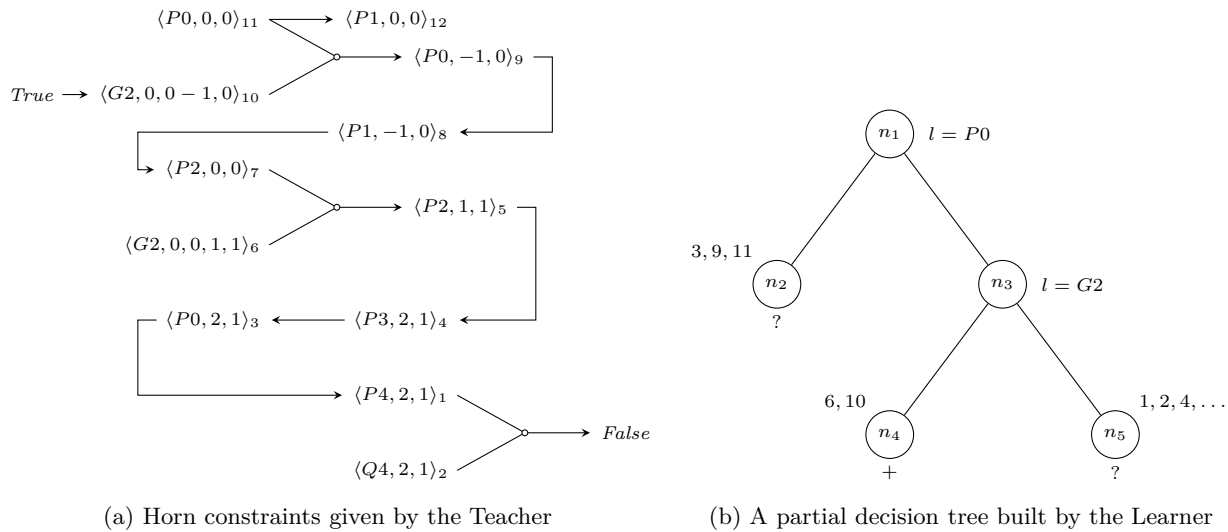
\begin{figure}[t]
	\centering

	\begin{subfigure}[b]{.53\textwidth}

		\begin{tikzpicture}[font=\scriptsize]

			\node[anchor=east] (true) at (-2.75, -1) {$\mathit{True}$};
			\node[anchor=east] (0) at (0, 0) {$\langle P0, 0, 0 \rangle_{11}$};
			\node[anchor=east] (1) at (0, -1) {$\langle G2, 0, 0 -1, 0 \rangle_{10}$};
			\node[anchor=east] (2) at (0, -2) {$\langle P2, 0, 0 \rangle_{7}$};
			\node[anchor=east] (3) at (0, -3) {$\langle G2, 0, 0, 1, 1 \rangle_{6}$};
			\node[anchor=east] (4) at (0, -4) {$\langle P0, 2, 1 \rangle_{3}$};
			\node[anchor=east] (5) at (2.5, 0) {$\langle P1, 0, 0 \rangle_{12}$};
			\node[anchor=east] (6) at (2.5, -1.5) {$\langle P1, -1, 0 \rangle_{8}$};
			\node[anchor=east] (7) at (2.5, -4) {$\langle P3, 2, 1 \rangle_{4}$};
			\node[anchor=east] (8) at (2.5, -5) {$\langle P4, 2, 1 \rangle_{1}$};
			\node[anchor=east] (9) at (2.5, -6) {$\langle Q4, 2, 1 \rangle_{2}$};

			\draw[->] (0) edge (5);
  		
			\draw[->] (0.east) -- ++(.9, -.5) coordinate (A) -- ++(.9, 0) node[anchor=west] (10) {$\langle P0, -1, 0 \rangle_{9}$};
			\draw (1.east) -- (A);
			\draw[fill=white] (A) circle (1.1pt);
		
		  	\draw[->] (10.east) -| ++(.25, -1) -- (6.east);

			\draw[<-, shorten <=1pt, shorten >= 0pt] (2.west) -| ++(-.25, .5) -- (6.west);
		
			\draw[->] (2.east) -- ++(.9, -.5) coordinate (B) -- ++(.9, 0) node[anchor=west] (11) {$\langle P2, 1, 1 \rangle_{5}$};
			\draw (3.east) -- (B);
			\draw[fill=white] (B) circle (1.1pt);
		
			\draw[->] (11.east) -| ++(.25, -1.5) -- (7.east);
	
			\draw[->] (7) edge (4);
	
			\draw[->] (4.west) -| ++ (-.25, -1) -- (8.west);
		
			\draw[->] (8.east) -- ++(.9, -.5) coordinate (C) -- ++(.9, 0) node[anchor=west] (false) {$\mathit{False}$};
			\draw (9.east) -- (C);
			\draw[fill=white] (C) circle (1.1pt);
		
			\draw[->] (true.east) -- (1.west);

		\end{tikzpicture}
		\caption{Horn constraints given by the Teacher} \label{fig:prog-hc} 
	\end{subfigure}
	\hfill
	\begin{subfigure}[b]{.45\textwidth}
		\begin{tikzpicture}[font=\scriptsize]
			\node[draw, shape=circle] (1) at (0, 0) {$n_1$};
			\node[draw, shape=circle] (2) at (-1.5, -2) {$n_2$};
			\node[draw, shape=circle] (3) at (1.5, -2) {$n_3$};
			\node[draw, shape=circle] (4) at (0, -4) {$n_4$};
			\node[draw, shape=circle] (5) at (3, -4) {$n_5$};
	
			\node[right=4mm] at (1) {$l = P0$};
			\node[right=4mm] at (3) {$l = G2$};

			\node[above left=2mm] at (2) {$3,9,11$};
			\node[below=3mm] at (2) {$?$};
	
			\node[above left=2mm] at (4) {$6,10$};
			\node[below=3mm] at (4) {$+$};
	
			\node[above right=2mm] at (5) {$1,2,4, \ldots$};
			\node[below=3mm] at (5) {$?$};
	
			\draw[shorten >=0pt] (1) -- (2) (1) -- (3) (3) -- (4) (3) -- (5);
	
		\end{tikzpicture}
		\caption{A partial decision tree built by the Learner} \label{fig:prog-dec-tree} 
	\end{subfigure}

\caption{Intermediate results of our Horn-ICE framework on the introductory example.} \label{fig:intermediate-results-example}
\end{figure}

The above example illustrates some of the key differences and challenges
from the classical ICE framework in \cite{LMN14,GNMR16}.
To begin with, the Learner needs to synthesize \emph{multiple}
invariants (in this case 12 different invariants) rather than a single
invariant; adapting decision-tree learning to learn multiple formulas
efficiently is challenging.
Secondly, we note that the Teacher \emph{needs} to use \emph{Horn}
counter-examples rather than simple implications, as illustrated by
the first counter-example $d_1 \And d_2 \Implies \false$, and
subsequently $d_6 \And d_7 \Implies d_5$ and $d_{10} \And d_{11}
\Implies d_{9}$.
Accordingly, the Learner's decision tree based algorithm needs to be
adapted to work robustly in the presence of Horn constraints.



\section{Decision Tree Learning with Horn Constraints}
\label{sec:learner}

\subsubsection*{Valuations and Horn constraints}
We will consider propositional formulas over a fixed set of
propositional variables $X$, using the usual
Boolean connectives $\notof$, $\And$, $\Or$, $\Implies$, etc.
The \emph{data points} we introduce later will also play the role of
propositional variables.
A \emph{valuation} for $X$ is a map $v : X \rightarrow \{\true,
\false\}$.
A given formula over $X$ evaluates, in the standard way, to
either $\true$ or $\false$ under a given valuation.
We say a valuation $v$ \emph{satisfies} a 
formula $\varphi$ over $X$, written $v \models \varphi$, if
$\varphi$ evaluates to $\true$ under $v$.

A \emph{partial} valuation for $X$ is a partial map
$u: X \pmap \{\true,\false\}$.
We denote by $\dom_{\true}(u)$  the set $\{x \in X \ | \ u(x) =
\true\}$ and $\dom_{\false}(u)$  the set $\{x \in X \ | \ u(x) =
\false\}$.
We say a partial valuation $u$ is \emph{consistent} with a formula
$\varphi$ over $X$, if there exists a full valuation $v$ for $X$,
which extends $u$ (in that for each $x \in X$, $u(x) = v(x)$
whenever $u$ is defined on $x$), and $v \models \varphi$.

Let $\varphi$ be a formula over $X$, and $u$ a partial valuation
over $X$ which is consistent with $\varphi$.
We say a variable $x \in X$ is \emph{forced to be true}, with
respect to $\varphi$ and $u$, if for all valuations $v$ which extend $u$,
whenever $v\models\varphi$ we have $v(x) = \true$.
Similarly we say $x$ is \emph{forced to be false}, with
respect to $\varphi$ and $u$, if for all valuations $v$ which extend $u$,
whenever $v\models\varphi$ we have $v(x) = \false$.
We denote the set of variables forced true (wrt $\varphi$ and $u$)
by $\forcedt{\varphi}{u}$, and those forced false by
$\forcedf{\varphi}{u}$.

Finally, for a partial valuation $u$ over $X$, and
subsets $T$ and $F$ of $X$, which are disjoint from each other and
from the domain of $u$, we denote by $u^T_F$ the partial valuation
which extends $u$ by mapping all variables in $T$ to $\true$, and
all variables in $F$ to $\false$.

A \emph{Horn clause} (or a \emph{Horn
  constraint}) over $X$ is disjunction of literals over $X$ with \emph{at
most} one positive literal.
Without loss of generality, we will write Horn clauses in one of the
three forms:
\begin{enumerate*}[label={(\arabic*)}]
	\item $\true \Implies x$,
	\item $(x_1 \And \cdots \And x_k) \Implies \false$, or
	\item $(x_1 \And \cdots \And x_l) \Implies y$,
\end{enumerate*}
where $l \geq 1$ and each of the $x_i$'s and $y$ belong to $X$.

\subsubsection*{Data Points} 
Our decision tree learning algorithm is paired with a teacher that
refutes incorrect conjectures with positive data points, negative
data points, and, more generally, with Horn constraints
over data points. Roughly speaking, a data
point corresponds to a program configuration and contains the values
of each program variable and potentially values that are derived from
the program variables, such as $x+y$, $x^2$, or
$\mathit{is\_list}(z)$. For the sake of a simpler presentation,
however, we assume that a data point is an element $d \in \mathbb D$
of some (potentially infinite) abstract domain of data points $\mathbb D$ (we
encourage the reader to think of programs over integers where data
points correspond to vectors of integers).

\subsubsection*{Base Predicates and Decision Trees}
The aim of our learning algorithm is to construct a decision tree
representing a Boolean combination of some base predicates.
We assume a set of base predicates, each of which evaluates to $\true$
or $\false$ on a data point $d \in D$.
More precisely, a \emph{decision tree} is a binary
tree $\T$ whose nodes either have two children (\emph{internal nodes}) or no
children (\emph{leaf nodes}), whose internal nodes are labeled with base
predicates, and whose leaf nodes are labeled with $\true$ or
$\false$. The formula $\psi_{\T}$ corresponding to a decision tree
$\T$ is defined to be $\bigvee_{\pi \in \Pi_\true} \bigl( \bigwedge_{\rho \in \pi} \rho \bigr)$ 
where $\Pi_\true$ is the set of all paths from the root of $\T$
to a leaf node labeled $\true$, and $\rho \in \pi$
denotes that the base predicate $\rho$ occurs as a label of a node on
the path $\pi$.
Given a set of data points $X \subseteq D$, a decision tree $\T$
induces a valuation $v_{\T}$ for $X$ given by $v_{\T}(d) = \true$ iff
$d \models \psi_{\T}$.
Finally, given a set of Horn constraints $C$ over a set of data
points $X$, we say a decision tree $\T$ is \emph{consistent} with $C$
if $v_{\T} \models \bigwedge C$.
We will also deal with ``partial'' decision trees, where some of the
leaf nodes are yet unlabeled, and we define a partial valuation
$u_{\T}$ corresponding to such a tree $\T$, and the notion of
consistency, in the expected way.

\subsubsection*{Horn Samples}
In the traditional setting, the learning algorithm collects the
information returned by the
teacher as a set of \emph{samples}, comprising ``positive'' and
``negative'' data points.
In our case, the set of samples will take the form of a set of Horn
constraints $C$ over a finite set of data points $X$.
We note that a positive point $d$ (resp.\@ a negative point $e$) can be
represented as a Horn constraint $\true \Implies d$ (resp.\@ $e \Implies
\false$).
We call such a pair $(X,C)$ a \emph{Horn sample}.

In each iteration of the learning process, we require the learning
algorithm to construct a decision tree $\T$ that agrees with the
information in the current Horn sample $\S = (X, C)$, in
that $\T$ is consistent with $C$.
The learning task we address then is
\textit{``given a Horn sample $\S$, construct a decision tree
	consistent with $\S$''}.

\begin{algorithm}[t]
\caption{Decision Tree Learner for Horn Samples}\label{alg:decision-tree}
\begin{algorithmic}[1]
\Procedure{Decision-Tree-Horn}{}
\ \ \ \ \ \ \ \ \ \ \Require A Horn sample $(X,C)$ 
\ \ \ \ \ \ \ \ \ \ \Ensure {A decision tree $\T$ consistent with $C$,
  if one exists.}
\State Initialize tree $\T$ with root node $r$, with $r.dat \gets X$
\State Initialize partial valuation $u$ for $X$, with $u \gets \emptyset$
\While {($\exists$ an unlabelled leaf in $\T$)}
  \State $n \gets \Call{Select-Node}{\,}$
  \If {(pure($n$))}  \ \ // $n.dat \intersection \dom_{\true}(u) =
  \emptyset$ or $n.dat \intersection \dom_{\false}(u) = \emptyset$
    \State $result = \Call{Label}{n}$
  \EndIf
  \If {($\notof \mbox{pure}(n) \Or \notof result$)}
    \If {($n.dat$ is singleton)}
       \State {print ``Unable to construct decision tree''}; 
       \Return \label{line:dec-tree-unsuccess-1} 
    \EndIf
    \State $result \gets \Call{Split-Node}{n}$
    \If {($\notof result$)}
       \State {print ``Unable to construct decision tree''};
       \Return \label{line:dec-tree-unsuccess-2} 
    \EndIf
  \EndIf
\EndWhile
\State \Return $\T$ \ \ \ \ {// decision tree constructed
  successfully} \label{line:dec-tree-success}
\EndProcedure
\end{algorithmic}
\begin{minipage}[t]{6.5cm}
\begin{algorithmic}[1]
\Procedure {Label}{node $n$}
\State $Y \gets n.dat \setminus dom(u)$;
\If {($n.dat \intersection \dom_{\true}(u) \neq \emptyset$)}  
  \Statex \ \ \ \ \ \ \ \ \ // $n$ contains only pos/unsigned pts
  \State $(res,T,F) \gets \Call{Horn-Solve}{u, Y, \emptyset}$
  \If {($res$)}
    \State $n.label \gets \true$
    \State $u \gets u^{Y \union T}_{F}$
    \State \Return $\true$
  \Else \ \Return $\false$;
  \EndIf
\Else \ \ldots\ \  // try to label neg
\EndIf
\EndProcedure
\end{algorithmic}
\end{minipage}
\begin{minipage}[t]{6cm}
\begin{algorithmic}[1]
\Procedure {Split-Node}{node $n$}
\State $(res, a) \gets \Call{Select-Attribute}{n}$
\If {($res$)}
  \State Create new nodes $l$ and $r$
  \State $l.dat \gets \{ d \in n.dat \ | \ d \models a \}$
  \State $r.dat \gets \{ d \in n.dat \ | \ d \not\models a \}$
  \State $n.\mathit{left} \gets l$, $n.\mathit{right} \gets r$.
  \State \Return $\true$
\Else 
  \State \Return $\false$
\EndIf
\EndProcedure
\end{algorithmic}
\end{minipage}
\end{algorithm}

\subsubsection*{The Learning Algorithm}
Our learning algorithm, shown in Algo.~\ref{alg:decision-tree},
is an extension of Garg et al.'s \cite{GNMR16}
learning algorithm, which in turn is based on the classical decision
tree learning algorithm of Quinlan \cite{Quinlan86}.
Given a Horn sample $(X,C)$, the algorithm creates an initial
(partial) decision tree $\T$ which has a single unlabelled node, whose
associated set of data points is $X$.
As an auxiliary data structure, the learning algorithm maintains a
partial valuation 
$u$, which is always an extension of the partial valuation induced by
the decision tree.
In each step, the algorithm picks an unlabelled leaf node $n$, and
checks if it is ``pure'' in that all points in the node are either
positive (i.e.\@ \true) or unsigned, or similarly neg/unsigned.
If so, it calls the procedure \textsc{Label} which
tries to label it positive if all points
are positive or unsigned, or negative otherwise.
To label it positive, the procedure first checks whether extending the
current partial valuation $u$ by making all the unsigned data points
in $n$ true results in a valuation that is consistent with the given
set of Horn constraints $C$.
It does this by calling the procedure \textsc{Horn-Solve} (described
in the next section), which not only checks whether the proposed
extension is consistent with $C$, but if so, also returns the set of
points forced to be true and false, respectively.
In this case, the node $n$ is labelled $\true$ and the partial
valuation $u$ is extended with the forced values.
If the attempt to label positive fails, it tries to label the node
negative, in a similar way.
If both these fail, it ``splits'' the node using a suitably chosen
base predicate $a$. The corresponding method \textsc{Select-Attribute},
which aims to (heuristically) obtain a small tree (i.e., a concise formula),
is described in Sec.~\ref{sec:attribute-selection}.

The crucial property of our learning algorithm is that if the given
set of constraints $C$ is satisfiable and if the data points in $X$ are
``separable,'' it will always construct a decision tree
consistent with $C$. We say that the points in $X$ are
\emph{separable} if for every pair of points $d_1$ and $d_2$ in $X$ we
have a base predicate $\rho$ which distinguishes them (i.e., either
$d_1 \models \rho$ and $d_2 \not\models \rho$, or vice-versa). This result, together with its time complexity,
is formalized in Theorem~\ref{thm:dec-tree-succ}.

\begin{theorem}
\label{thm:dec-tree-succ}
Let $(X, C)$ be a Horn sample, $n = |X|$, and  $h = |C|$.
If the input set of points $X$ is separable and the input Horn
constraints $C$ are satisfiable, then Algo.~\ref{alg:decision-tree}
runs in time $\mathcal O(h\cdot n^3)$ and returns a decision tree that is consistent with the Horn sample
$(X,C)$.
\end{theorem}

\begin{proof}
At each iteration step, the algorithm maintains the invariant that the
partial valuation $u$ is an extension of the partial valuation
$u_{\T}$ induced by the current (partial) decision tree $\T$, and is
consistent with $C$. This is because each labelling step is first
checked by a call to the horn-solver, which also correctly identifies
the set of forced values, which are then used to update $u$.
It follows that if the algorithm terminates successfully in
Line~\ref{line:dec-tree-success}, then $u_{\T}$ is a full valuation
which coincides with $u$, and hence satisfies $C$.
The only way the algorithm can terminate unsuccessfully is in
Line~\ref{line:dec-tree-unsuccess-1} or
Line~\ref{line:dec-tree-unsuccess-2}.
The first case is ruled out since if $n.dat$ is singleton, and by
assumption $u_{\T}$ is consistent with $C$, we must be able to label
the single data point with either $\true$ or $\false$ in a way that is
consistent with $C$.
The second case is ruled out, since under the assumption of
separability the \textsc{Select-Attribute} procedure will always return a
non-trivial attribute (see Sec.~\ref{sec:attribute-selection}).

The learning algorithm (Algorithm~\ref{alg:decision-tree}) runs in cubic time
in the size $n$ of the input set $X$ of data points and linear in the number $h$ of Horn constraints.
To see this, observe that in each iteration of the loop the algorithm
produces a tree that is either the same as the previous step (but with a
leaf node labelled), or splits a leaf to extend the previous tree.
At each step we maintain an invariant that the collection of data points
in the leaf nodes forms a partition of the input set $X$. Thus the
number of leaf nodes is bounded by $n$, and hence each tree has a
total of at most $2n$ nodes. When the algorithm returns (successfully or
unsucessfully) each node in the final tree has been processed at
most once by calling the labelling/splitting subroutines on
it. Furthermore, the main work in the subroutines is the call to the
horn-solver, which takes $\mathcal O(h\cdot n^2)$ time (see
Sec.~\ref{sec:horn-solver}). It follows that
Algorithm.~\ref{alg:decision-tree} runs in $\mathcal O(h\cdot n^3)$ time. \qed
\end{proof}

When the points in $X$ are not separable, as done in \cite{GNMR16}, we
add ``iff'' constraints between every pair of inseparable points, and
if the resulting Horn constraints are satisfiable, our algorithm is
guaranteed to construct a decision tree consistent with the given
Horn constraints.

Furthermore, we can extend our algorithm to work on an infinite enumerable set
of predicates ${\cal P}$, and assure that the algorithm will find an invariant
if there is one, as done in~\cite{GNMR16}. We can  take some finite set $X \subseteq P$,
asking whether there is some invariant over $X$ that satisfies $P$, and if not
grow $X$ by taking finitely more predicates from ${\cal P}\setminus X$. This is clearly guaranteed to converge
on an invariant if one is expressible over ${\cal P}$. 



\section{Algorithm for solving and propagating Horn constraints}
\label{sec:horn-solver}

In the decision-tree based learning approach, our aim is to construct a
decision-tree representing a Boolean combination of some base
predicates, that is consistent with a given set of Horn constraints.
A crucial step in our decision-tree algorithm is to check whether a proposed
extension of the current partial valuation maintained by the Learner,
is indeed consistent with
the given Horn constraints, and if so to determine the set of
propositional variables that are ``forced'' to be true or false.
In this section we describe an efficient algorithm to carry out
this step.
Our algorithm is an adaptation of the ``pebbling'' algorithm of
\cite{DowlingG84} for checking satisfiability of Horn formulas, to
additionally find the variables that are forced to be true or false.

\begin{algorithm}[th]
\caption{}\label{alg:horn-solver}
\begin{algorithmic}[1]
\Procedure{Horn-Solve}{}
\ \ \ \ \ \ \ \ \ \ \Require Horn constraints $C$ over $X$, partial
valuation $u$ over $X$, and $T, F \subseteq X$.
\ \ \ \ \ \ \ \ \ \ \Ensure {``Unsat'' if $u^T_F$ is inconsistent with $C$; 
        ``Sat'', \forcedt{C}{u^T_F}, \forcedf{C}{u^T_F} otherwise.}
\State Add two new variables $\True$ and $\False$ to $X$. Let $X'
 = X \union \{\True, \False\}$.
\State $C'$ $\gets$ $C$ +
clauses $\true \Implies x$ for each $x$ such
that $u(x) = \true$ or $x \in T$, and $x \Implies \false$ for each $x$
such that $u(x) = \false$ or $x \in F$.
\State Mark variable $\True$ with ``$\mymark$'', and
each variable $x \in X$ with ``$\mymark x$''.
\Repeat {
\State For each constraint $x_1 \And \cdots \And x_l \Implies y$
in $C'$:
\If {$x_1$, \ldots, $x_l$ are all marked ``$\mymark$''}
       {mark $y$ with ``$\mymark$''}
       \EndIf
\If {$\exists z \in X$ s.t.\@ $x_1$, \ldots,
  $x_l$ are all marked ``$\mymark z$'' or ``$\mymark$''}
       {mark $y$ with ``$\mymark z$''}
       \EndIf
}
\Until{no new marks can be added}
\State $P \gets \{ x \in X \ | \ x \mbox{ is marked }
``\mymark\mbox{''} \}$,
       $N \gets \{ x \in X \ | \ \False \mbox{ is marked }
``\mymark\!x\mbox{''} \}$
\State $P' \gets P - (T \union \dom_{\true}(u))$,
       $N' \gets N - (F \union \dom_{\false}(u))$
\If {($\False$ is marked ``\mymark'')} \Return ``Unsat''
\Else{} \Return ``Sat'', $P'$, $N'$
\EndIf
\EndProcedure
\end{algorithmic}
\end{algorithm}

Procedure~\textsc{Horn-Solve} in Algorithm~\ref{alg:horn-solver}
shows our procedure for checking consistency of a partial
valuation with respect to a given set of Horn constraints $C$, as well
as identifying the subset of variables forced to true or false.
Intuitively, the standard linear-time algorithm for Horn satisfiability
in fact already identifies the \emph{minimal} $M$ set of variables that are
forced to be true in any satisfying valuation, and assures us that the
others $\vec{M}$ can be set to false. However, the other variables are not
\emph{forced} to be false. Our algorithm essentially runs another
set of SAT problems where each of the other variables are set to \emph{true}
(this is modeled by the variable being marked $*$ in the algorithm);
this returns SAT iff the variable is not forced to be \emph{false}.
The following example illustrates Algorithm~\ref{alg:horn-solver}.

\begin{figure}[th]
	 \centering
	\begin{tikzpicture}

		\node (0) at (0, 0) {$y$};
		\node (1) at (0, -1.4) {$x$};
		\draw[->] (0.east) -- ++(.9, -.7) coordinate (A) -- ++(.9, 0) node[anchor=west] (2) {$z$};
		\draw (1.east) -- (A);
		\draw[fill=white] (A) circle (1.1pt);
		\draw[->] (1.north) -- (0.south);
		\draw[shorten <=1pt, shorten >=0pt, <-] (1.west) -- ++(-.9, 0) node[anchor=east] (true) {$\True$};

		\node (3) at (4, 0) {$b$};
		\node (4) at (4, -1.4) {$a$};
		\draw[->] (3.east) -- ++(.9, -.7) coordinate (B) -- ++(.9, 0) node[anchor=west] (false) {$\False$};
		\draw (4.east) -- (B);
		\draw[fill=white] (B) circle (1.1pt);
		\draw[->] (3) edge[bend left] (4);
		\draw[->] (4) edge[bend left] (3);

		\begin{scope}[font=\scriptsize, inner sep=0pt, text=black!75]
			\node[anchor=north, yshift=-.25ex] at (true.south) {\strut $\ast$};
			\node[anchor=south west, xshift=-.25ex, yshift=-.25ex] at (0.north east) {\strut $\ast, \ast y, \ast x$};
			\node[anchor=south east, xshift=.25ex, yshift=-.25ex] at (0.north west) {\strut (+)};
			\node[anchor=north west, xshift=-.25ex, yshift=.25ex] at (1.south east) {\strut$\ast, \ast x$};
			\node[anchor=north east, xshift=.25ex, yshift=.25ex] at (1.south west) {\strut +};
			\node[anchor=north, yshift=.25ex] at (2.south) {\strut$\ast, \ast z, \ast x, \ast y$};
			\node[anchor=south, yshift=-.25ex] at (2.north) {\strut (+)};
			\node[anchor=south west, xshift=-.25ex, yshift=-.25ex] at (3.north east) {\strut $\ast b, \ast a$};
			\node[anchor=south east, xshift=.25ex, yshift=-.25ex] at (3.north west) {\strut (-)};
			\node[anchor=north west, xshift=-.25ex, yshift=.25ex] at (4.south east) {\strut $\ast a, \ast b$};
			\node[anchor=north east, xshift=.25ex, yshift=.25ex] at (4.south west) {\strut (-)};
			\node[anchor=north, yshift=-.25ex] at (false.south) {\strut $\ast a, \ast b$};
		\end{scope}

	\end{tikzpicture}

	 \caption{Example illustrating Algorithm~\ref{alg:horn-solver}. The given set of Horn constraints is $C = \{x \Implies y, x \And y  \Implies z, a \Implies b, b \Implies a, a \And b \Implies  \False\}$, $T = \{x\}$, $F = \emptyset$, and the partial valuation is empty.
 The algorithm outputs ``Sat'' together with $P = \{y,z\}$ and $N = \{a,b\}$.} \label{fig:horn-algo-example}
\end{figure}
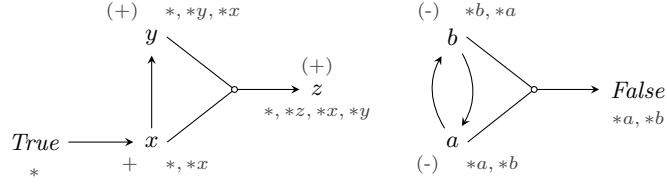

\begin{example}
Fig.~\ref{fig:horn-algo-example} illustrates the working of procedure \textsc{Horn-Solve} on an example set of Horn constraints $C = \{x \Implies y, x \And y \Implies z, a \Implies b, b \Implies a, a \And b \Implies \False\}$. The given partial valuation is empty, the set of variables $T$ initially set to $\true$ is $\{x\}$, and the set of variables $F$ set to $\false$ is $\emptyset$. 
The final marking computed by the procedure is shown below each variable. Variables set to $\true$ in the partial assignment or in the input set $T$, are shown with a ``+'' above them. Variables forced to true (respectively false) are shown with a ``(+)'' (respectively ``(-)'') above them. \exampleend
\end{example}

For the remainer of this section, let us fix $X$, $C$, $u$, $T$, and $F$ to be the inputs to the
procedure, and let $X'$, $C'$, $P$, $N$, $P'$ and $N'$ be as described
in the algorithm.
It is clear that there exists an extension of $u^T_F$ satisfying
$C$ ifand only if $C'$ is satisfiable.
Furthermore, the set of variables forced true
by $C$ with respect to $u^T_F$ coincides with those forced
true in $C'$, less the variables in $\dom_\true(u) \union T$.
A similar claim holds for the variables forced to false.
In total, we obtain the following result.

\begin{theorem}
\label{thm:horn-solver}
Let $(X, C)$ be a Horn sample, $n = |X|$, and  $h = |C|$. The procedure \textsc{Horn-Solve} runs in time $\mathcal O(h\cdot n^2)$ and outputs
\begin{itemize}
	\item ``Sat'' and the set of variables forced to be true (in $P$) and false (in $N$), respectively,
if the given extended valuation $u^T_F$ is consistent with the Horn constraints $C$; and
	\item ``Unsat'' otherwise.
\end{itemize}
\end{theorem}

To prove Theorem~\ref{thm:horn-solver},  we first introduce the notion of \emph{pebblings}  (adapted from \cite{DowlingG84}) and state several straight-forward propositions.
Let $x$ be a variable in $X'$.
A $C'$-\emph{pebbling} of $(x,m)$ from $\True$ is a
sequence of markings $(x_0,m_0), \ldots, (x_k,m_k)$, such that
$x_k = x$ and $m_k = m$, and each $x_i$ and $m_i$ satisfy:
\begin{itemize}
\item $x_i \in X'$ and
  $m_i \in \{*\} \union \{*x \ | \ x \in X \}$, and
\item one of the following:
  \begin{itemize}
  \item $x_i = \True$ and $m_i = ``*\mbox{''}$, or
  \item $m_i = ``*x_i\mbox{''}$, or
  \item $\exists i_1, \ldots, i_l < i$ such that each $m_{i_k} =
    ``*\mbox{''}$, and $m_i = ``*\mbox{''}$, and $x_{i_1} \And \cdots
    \And x_{i_l} \Implies x_i \in C'$, or
  \item $\exists z \in X$ and $\exists i_1, \ldots, i_l < i$ such
    that each $m_{i_k} = ``*\mbox{''}$  or ``$*z$'', and $m_i = ``*\!z\mbox{''}$,
    and $x_{i_1} \And \cdots \And x_{i_l} \Implies x_i \in C'$.
  \end{itemize}
\end{itemize}
A $C'$-pebbling is \emph{complete} if the sequence cannot be
extended to add a new mark.

It is easy to see that each time the procedure
\textsc{Horn-Solve} marks a variable $x$ with a mark $m$, the sequence
of markings till this point forms a valid $C'$-pebbling of $(x,m)$ 
from $\True$, and the final sequence of markings produced does
indeed constitute a complete pebbling.

\begin{proposition}
\label{prop:peb1}
Consider a $C'$-pebbling of $(x,``*\!\mbox{''})$ from $\True$.
Let $v$ be a valuation such that $v\models C'$.
Then $v(x) = \true$.
\end{proposition}
  
\begin{proposition}
\label{prop:peb4}
Consider a $C'$-pebbling of $(y,``*x\!\mbox{''})$ from $\True$.
Let $v$ be a valuation such that $v(x) = \true$ and $v\models C'$.
Then $v(y) = \true$.
\end{proposition}

\begin{proposition}
\label{prop:peb2}
Consider a complete $C'$-pebbling from $\True$, in which $\False$ is
not marked ``$*$''.
Let $v$ be the valuation given by
\[
v(x) = \begin{cases} \true & \text{if $x$ is marked ``$*$''; and} \\ \false & \text{otherwise.} \end{cases}
\]
Then $v \models C'$.
\end{proposition}

\begin{proposition}
\label{prop:peb3}
Let $y \in X$. 
Consider a complete $C'$-pebbling from $\True$, in which $\False$ is
not marked ``$*y$''.
Let $v$ be the valuation given by
\[
	v(x) = \begin{cases} \true & \text{if $x$ is marked ``$*$'' or ``$*y$''; and} \\ \false & \text{otherwise.} \end{cases}
\]
Then $v \models C'$.
\end{proposition}

Given Propositions~\ref{prop:peb1} to \ref{prop:peb3}, we can now prove Theorem~\ref{thm:horn-solver}.

\begin{proof}[of Theorem~\ref{thm:horn-solver}]
We first observe that the procedure \textsc{Horn-Solve} clearly terminates since there is a finite number
of variables in $X$ (say $n$) and each variable can only be
marked $n+1$ times.
In fact, if $h$ is the number of clauses in $C$, the running time
of the procedure is bounded by $\mathcal O(h\cdot n^2)$.

We return now to the correctness of the procedure \textsc{Horn-Solve}.
Firstly, the procedure outputs ``Unsat'' if and only if $C'$ is unsatisfiable.
Suppose the procedure outputs ``Unsat''.
Then, $\False$ must be marked ``$*$'' by the algo.
But this means we have a $C'$-pebbling of $\False$ from $\True$.
Now, since we don't have a clause $\True \Implies \False$ in $C'$, we
must have a clause $c \in C'$ of the form $x_1 \And \cdots \And x_l
\Implies \False$ and a pebbling (by ``$*$'') of each $x_i$ from
$\True$.
Now, suppose to the contrary that there exists a valuation $v$ such
that $v \models C'$.
By Prop.~\ref{prop:peb1}, $v$ must set each $x_i$ to $\true$. But this
means that $v$ does not satisfy clause $c$, which contradicts the
assumption that $v \models C'$.

Conversely, suppose the procedure outputs ``Sat''.
Then, we have have a complete pebbling in which
$\False$ is not labelled ``$*$''.
By Prop.~\ref{prop:peb2}, we can construct a valuation $v$ which
satisfies $C'$, and hence $C'$ is satisfiable.

We further need to argue that in the latter case (when the algorithm
outputs ``Sat''), the $P$ and $N$ sets are computed correctly as the
set of variables forced to true and false respectively, in
$C'$.
We begin with the $P$ set.
Recall that $P$ is the set of variables that are marked ``$*$'' at the
end of the procedure, and hence in a complete pebbling.
It follows from Prop.~\ref{prop:peb1} that all the variables in $P$
are indeed forced true.
Conversely, suppose there is a variable $x$ that is not marked ``$*$''
by the algorithm.
Then, by Prop.~\ref{prop:peb2} there is a valuation $v$ which satisfies
$C'$, in which $x$ is set to $\false$. Hence $x$ is not forced to
true by $C'$.

Regarding the $N$ set, recall that it is the set of variables $x$
such that $\False$ is marked ``$*x$''.
Let $x \in N$, and let $v$ be a valuation that sets $x$ to $\true$,
and satisfies $C'$.
Then there must be a clause $c$ of the form $x_1 \And \cdots \And x_l
\Implies \false$ in $C'$, and each $x_i$ is marked ``$*$'' or ``$*x$''.
By Props.~\ref{prop:peb1} and \ref{prop:peb4}, $v$ must set each of
the $x_i$'s to $\true$.
But then $v$ does not satisfy clause $c$, and this contradicts the
assumption that $v \models C'$. Hence every variable in $N$ is forced
false.
Conversely, suppose $x \not\in N$.
Hence $\False$ is not marked ``$*x$''.
Then, by Prop.~\ref{prop:peb3}, there is a valuation $v$ which sets $x$
to $\true$ and satisfies $C'$.
Hence, $x$ is not forced to $\false$ by $C'$.
This completes the proof of the correctness of the procedure.
\qed
\end{proof}



\section{Node and Attribute Selection}
\label{sec:attribute-selection}

The decision tree algorithm in Section~\ref{sec:learner} returns a consistent tree irrespective of the order in which nodes of the tree are processed or the heuristic used to choose the best attribute to split nodes in the tree. If one is not careful while selecting the next node to process or one ignores the Horn constraints while choosing the attribute to split the node, seemingly good splits can turn into bad ones as data points involved in the Horn constraints get classified during the construction of the tree.
We experimented with the following strategies for node and attribute selection:
 \begin{description}
    \item[Node selection:] breadth-first-search; depth-first-search; random selection; selecting nodes with the maximum/minimum entropy
    \item[Attribute selection:] based on a new information gain metric that penalizes node splits that cut Horn constraints; based on entropy for Horn samples obtained by assigning probabilistic likelihood values to unlabeled datapoints using model counting.
\end{description}

So as to clutter the paper not too much, we here only describe the best performing combination of strategies in detail. The experiments reported in Section~\ref{sec:experiments} have been conducted with this combination.

\subsubsection*{Choosing the next node to expand the decision tree}
We select nodes in a breadth-first search (BFS) order for building the decision tree. BFS ordering ensures that while learning multiple invariant annotations, the subtree for all invariants gets constructed simultaneously. In comparison, in depth-first ordering of the nodes, subtrees for the multiple invariants are constructed one after the other. In this case, learning a simple invariant, for e.g.\@ \emph{true}, for an annotation usually forces the invariant for a different annotation to become very complex.

\subsubsection*{Choosing attributes for splitting the given node}
Similar to ~\cite{GNMR16}, we observed that if one chooses attribute splits based on the entropy of the node that ignores Horn constraints, the invariant learning algorithm tends to produce large trees. 
In the same spirit as ~\cite{GNMR16}, we penalize the information gain for attribute splits that cut Horn constraints, and choose the attribute with the highest corresponding information gain. 
For a sample $S = (X, C)$ that is split with respect to attribute $a$ into subsamples $S_a$ and $S_{\neg a}$, we say that the corresponding attribute split cuts a Horn constraint $\psi \in C$ if and only if
\begin{itemize}
	\item $x \in \textit{premise}(\psi)$ and $x \in S_a$ and $\textit{conclusion}(\psi) \in S_{\neg a}$; or 
	\item $x \in \textit{premise}(\psi)$ and $x \in S_{\neg a}$ and $\textit{conclusion}(\psi) \in S_a$.
\end{itemize}
The \emph{penalized} information gain is defined as
\[ \textit{Gain}_\textit{pen}(S, S_a, S_{\neg a}) = \textit{Gain}(S, S_a, S_{\neg a}) - \textit{Penalty}(S, S_a, S_{\neg a}, C), \]
where the penalty is proportional to the number of Horn constraints cut by the attribute split. However, we do not penalize a split when it cuts a Horn constraint such that the premise of the constraint is labeled negative and the conclusion is labeled positive. We incorporate this in the penalty function by formally defining it as
\[ \textit{Penalty}(S, S_a, S_{\neg a}, H) = \sum\limits_{\substack{\psi \in H, x \in S_a\\x \in \textit{premise}(\psi) \\ \textit{conclusion}(\psi) \in S_{\neg a}}} \big( 1 - f(S_a, S_{\neg a})\big) + \sum\limits_{\substack{\psi \in H, x \in S_{\neg a} \\x \in \textit{premise}(\psi) \\ \textit{conclusion}(\psi) \in S_{a}}} \big( 1 - f(S_{\neg a}, S_{a})\big), \]
\noindent where, for subsamples $S_1$ and $S_2$, $f(S_1, S_2)$ is the likelihood of $S_1$ being labeled negative and $S_2$ being labeled positive (i.e., $f(S_1, S_2) = \frac{N_1}{P_1 + N_1}.\frac{P_2}{P_2 + N_2}$). Here, $P_i$ and $N_i$ is the number of positive and negative datapoints respectively in the sample $S_i$.


\section{Experimental Evaluation}
\label{sec:experiments}
We have implemented a prototype, named \emph{Horn-DT} to evaluate the proposed learning framework.
The decision tree learning algorithm and the Horn solver are fresh implementations, consisting of roughly 6000 lines of C\raisebox{1pt}{\small {+\!+}} code. The teacher, on the other hand, is implemented on top of Microsoft's program verifier Boogie~\cite{DBLP:conf/fmco/BarnettCDJL05} and reuses much of the code originally developed by Garg et al.~\cite{GNMR16} for their ICE learning tool. Via C-to-Boogie conversion tools such as \textsc{Smack}~\cite{DBLP:conf/cav/RakamaricE14}, our prototype becomes a fully automatic program verifier for C programs.

We have evaluated our prototype on two suits of benchmarks. The first suite consists of 56 sequential programs of varying complexity, mainly taken the Software Verification Competition (\textsc{SV-Comp})~\cite{DBLP:conf/tacas/Beyer17}. 
The second benchmark suite consists of 24 concurrent and recursive programs and includes popular concurrent protocols like Peterson's algorithm, a producer-consumer problem and complex recursive programs like McCarthy91.

A natural choice to compare Horn-DT to is \emph{Seahorn}~\cite{DBLP:conf/cav/GurfinkelKKN15}. Seahorn is a fully automated analysis framework for LLVM-based languages, which compiles C and C\raisebox{1pt}{\small {+\!+}} code via the LLVM intermediate language into verification conditions in form of constrained Horn clauses. These Horn clauses can then be solved by any suitable solver. For our experimental evaluation, we used the PDR engine~\cite{DBLP:conf/sat/HoderB12} implemented in Z3~\cite{DBLP:conf/tacas/MouraB08}, which is offered as the standard option in Seahorn.

Unfortunately, Seahorn cannot reason about concurrent programs and, thus, we can only compare to Seahorn on sequential programs (we are not aware of any other Horn-based invariant synthesis tool to which we could compare on this benchmark suite).
Moreover, it is important to emphasize that our evaluation is not a like-to-like comparison due to the following two reasons:
\begin{itemize}
	\item Boogie and Seahorn produce different verification conditions: as Boogie allows very rich specifications (e.g., including quantifiers), the verification conditions produced by it are typically larger and more complex compared to those produced by Seahorn.
	\item The PDR engine does not require invariant templates while Horn-DT does. More precisely, Horn-DT learns invariants that are arbitrary Boolean combinations of Boolean program variables and predicates of the form $x \pm y \leq c$ where $x,y$ are numeric program variables and $c$ is a constant determined by the decision tree learner.	(Note, however, that this particular choice of template is not a general restriction of our framework and can easily be changed should the specific domain require it.)
\end{itemize}
As a consequence of these differences, our general expectation was that Seahorn with PDR to performs slightly better than Horn-DT on ``simple'' programs that do not require the overhead introduced by Boogie and permit simple invariants. Conversely, we expected a superior performance of Horn-DT on complex programs (e.g., programs working over arrays or performing nonlinear computations) due to its black box nature and the use of templates.

\pgfplotsset{runtime diagram/.style={%
	xmin = 1e-2, xmax = 1e2,%
	ymin = 1e-2, ymax = 1e2,%
	log basis x=10,%
	log basis y=10,%
	enlarge x limits=false, enlarge y limits=false,%
	xtickten={-1,...,1},%
	ytickten={-1,...,1},%
	extra x ticks={1e2}, extra x tick labels={\strut TO},%
	extra y ticks={1e2}, extra y tick labels={\strut TO},%
	xlabel near ticks,
	ylabel near ticks,
}}

\pgfplotsset{runtime diagram marks/.style={%
	only marks,%
	mark=x,%
	mark size=2.5,%
}}

 \begin{figure}[th]
	\centering

	\begin{tikzpicture}
		\begin{axis}[
			name=plot1,
			width=50mm,
			y = 7.5mm,
			xbar stacked,
			xmin = 0, xmax = 56,
			extra x ticks={56},
			enlarge y limits={.75},
			xlabel={Number of programs},
			symbolic y coords={Horn-DT, Seahorn},
			ytick=data,
			legend style={at={(0,-0.6)}, anchor=north west, legend columns=1, draw=none, legend cell align=left, column sep=.75em},
			cycle list={{fill=black!60,draw=black}, {fill=black!35,draw=black}, {fill=black!15,draw=black}},
			]
			
			\addplot+ [xbar] coordinates {(56,Horn-DT) (43,Seahorn) };
			\addlegendentry{Verified}
			
			\addplot+ [xbar] coordinates {(0,Horn-DT) (3,Seahorn) };
			\addlegendentry{Timeout ($300$ s)}
						
			\addplot+ [xbar] coordinates {(0,Horn-DT) (10,Seahorn) };
			\addlegendentry{Error / false positive}
			
		\end{axis}

		\begin{loglogaxis}[
			at={(plot1.north east)}, anchor=north west, xshift=40mm,
			runtime diagram,
			height=50mm,
			xlabel = {Seahorn (time in s)},
			ylabel = {Horn-DT (time in s)},
		]
	
		\addplot[runtime diagram marks] table[col sep=comma, x={Seahorn-Z3}, y={Horn-DT}] {horn_vs_seahorn.csv};
	
		\draw[black!25] (rel axis cs:0, 0) -- (rel axis cs:1, 1);
	
		\end{loglogaxis}
	\end{tikzpicture}

	\caption{Experimental comparison of Horn-DT with Seahorn on the the sequential programs benchmark suite. TO indicates a time out after 300 s or a false positive.} \label{fig:experimental-results-sequential}

\end{figure}
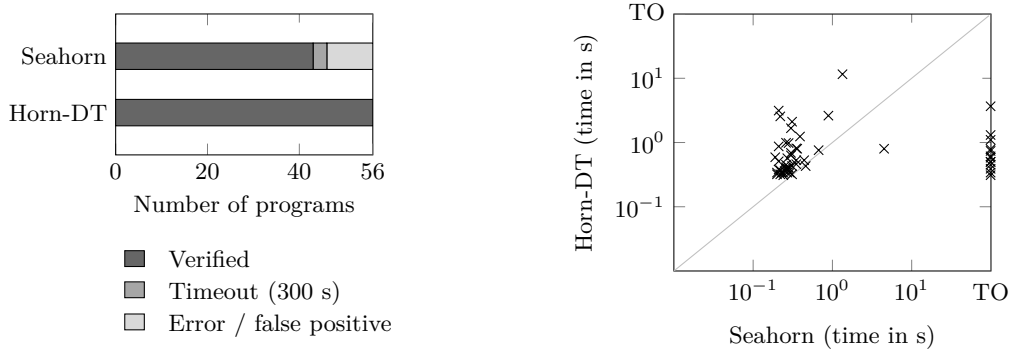

The remainder of this section presents our empirical evaluation on both benchmark suites in detail. All experiments were conducted on a Intel Core i3-4005U 4x1.70GHz CPU with 4\,GB of RAM running Ubuntu 16.04 LTS 64\,bit.
The reported results were obtained using the split and node selection strategies described in Sec.~\ref{sec:attribute-selection}.

\subsubsection*{Sequential Programs}
The first benchmark suite consists of 56 sequential programs taken \textsc{SV-Comp} as well as Garg et al.'s ICE tool~\cite{GNMR16}. These programs vary in complexity and range from simple integer manipulating programs to programs involving nonlinear computations to programs with complex array manipulations. The task for each of the benchmark programs is to find a single inductive invariant or adequate contract for a method call. We refer the reader to \textsc{SV-Comp}~\cite{DBLP:conf/tacas/Beyer17} as well as Garg et al.~\cite{GNMR16} for further details.

Figure~\ref{fig:experimental-results-sequential} summarizes the results of our experimental evaluation on the sequential programs benchmark suite (see Appendix~\ref{app:seq-experiments-details} for details). As shown on the left, Horn-DT is able to synthesize invariants and verify all 56 programs in the benchmarks suite. Seahorn, on the other hand, is able to verify 43 programs (77\,\%) and times out on 3 further programs (5\,\%). Surprisingly, Seahorn produced false positives (i.e., it reported an assertion violation) for 10 programs (18\,\%)---these programs predominantly involve arrays or nonlinear computations. After corresponding with the authors of SeaHorn, it turned out that these false positives were in fact caused by bugs in Seahorn and are now being investigated.

The right-hand-side of Figure~\ref{fig:experimental-results-sequential} shows a comparison of the runtimes of Horn-DT and Seahorn. As expected, Seahorn requires slightly less time to verify programs but frequently times out (or fails) when the programs work over arrays or perform nonlinear computations. Horn-DT, on the other hand, is able to verify the majority of programs in less than 1 s, which demonstrates the effectiveness of our technique.

\subsubsection*{Concurrent and Recursive Programs}
The second benchmark suite consists of 17 concurrent and recursive
programs selected from \textsc{SV-Comp15} \emph{pthread},
\emph{pthread-atomic}, \emph{pthread-ext}, \emph{pthread-lit},
\emph{recursive} and \emph{recursive-simple} benchmarks, as well as
popular concurrent programs from the literature~\cite{mine2014relational}.

\begin{table}[t]
\centering
\caption{Results of Horn-DT on concurrent and recursive programs. Columns show the number of invariants to be synthesized (``Inv''), total number of terms used (``Dim''), number of iterations between teacher and learner (``Rounds''), number of each kind of counter-examples generated (``Pos'', ``Neg'', and ``Horn''), and the time taken (``Time''). Benchmarks with the suffix \emph{-RG} and \emph{-OG} indicate Rely-Guarantee-style proofs and Owicki-Gries-style proofs, respectively.} \label{table:conc-rec-results}

\bigskip
\begin{tabular}{lrrrrrrr}
\toprule
\multicolumn{1}{l}{Benchmark} 				& \multicolumn{1}{c}{Inv} & \multicolumn{1}{c}{Dim} & \multicolumn{1}{c}{Rounds} & \multicolumn{1}{c}{Pos} & \multicolumn{1}{c}{Neg} & \multicolumn{1}{c}{Horn} & \multicolumn{1}{c}{Time (s)} \\
\midrule
Ackermann02                           				& 2                       & 5                       & 13                         & 4                       & 1                       & 8                        & 0.89                         \\ 
Addition03                            				& 2                       & 5                       & 31                         & 2                       & 5                       & 29                       & 1.00                         \\ 
afterrec\_2calls                      				& 2                       & 4                       & 0                          & 0                       & 0                       & 0                        & 0.04                         \\ 
BallRajamani-SPIN2000-Fig1            				& 2                       & 5                       & 32                         & 5                       & 0                       & 30                       & 1.58                         \\ 
fib\_bench-OG                         				& 6                       & 18                      & 7                          & 2                       & 0                       & 5                        & 1.52                         \\ 
fibo\_2calls\_5                       				& 4                       & 6                       & 105                        & 1                       & 3                       & 105                      & 3.65                         \\ 
Fibo\_5                               				& 2                       & 3                       & 71                         & 1                       & 3                       & 67                       & 2.43                         \\ 
id2\_b3\_o2                           				& 4                       & 6                       & 48                         & 1                       & 5                       & 45                       & 1.40                         \\ 
id2\_i5\_o5                           				& 4                       & 6                       & 10                         & 1                       & 1                       & 10                       & 0.72                         \\ 
McCarthy91                            				& 1                       & 2                       & 98                         & 5                       & 91                      & 2                        & 2.77                         \\ 
peterson-OG                           				& 8                       & 32                      & 370                        & 3                       & 91                      & 277                      & 65.31                        \\ 
qw2004-OG                             				& 13                      & 20                      & 24                         & 2                       & 1                       & 23                       & 3.33                         \\ 
stateful01-OG                         				& 6                       & 12                      & 220                        & 2                       & 0                       & 218                      & 9.56                         \\ 
sum\_2x3                              				& 2                       & 5                       & 36                         & 1                       & 2                       & 36                       & 0.99                         \\ 
sum\_non\_eq                          				& 2                       & 4                       & 8                          & 3                       & 1                       & 4                        & 0.58                         \\ 
sum\_non                              				& 2                       & 4                       & 6                          & 4                       & 1                       & 1                        & 0.57                         \\ 
18\_read\_write\_lock-OG              				& 8                       & 16                      & 16                         & 2                       & 1                       & 13                       & 2.42                         \\
\midrule
Mine\_Fig\_1-OG~\cite{mine2014relational}                       & 10                      & 20                      & 35                         & 2                       & 0                       & 41                       & 2.00                         \\ 
Mine\_Fig\_1-RG~\cite{mine2014relational}			& 12                      & 28                      & 43                         & 2                       & 0                       & 48                       & 3.13                         \\ 
Mine\_Fig\_4-OG~\cite{mine2014relational}			& 13                      & 28                      & 51                         & 3                       & 0                       & 65                       & 6.03                         \\ 
Mine\_Fig\_4-RG~\cite{mine2014relational}			& 15                      & 44                      & 123                        & 3                       & 0                       & 159                      & 38.59                        \\ 
pro\_cons\_atomic-OG                  				& 8                       & 16                      & 32                         & 2                       & 0                       & 30                       & 2.11                         \\ 
pro\_cons\_queue-OG                   				& 8                       & 16                      & 31                         & 2                       & 0                       & 29                       & 2.09                         \\ 
12\_hour\_clock\_serialization        				& 1                       & 3                       & 6                          & 1                       & 1                       & 4                        & 1.59                         \\
\bottomrule
\end{tabular}
\end{table}

For concurrent programs we have used both
Rely-Guarantee~\cite{XuRH97} and
Owicki-Gries~\cite{owicki1976verifying} proof techniques to
verify the assertions.  For recursive programs we use a modular
verification technique, in the form of function contracts for each
procedure.
All these programs were manually converted into Boogie
programs. For concurrent programs, we essentially encode the
verification conditions for Rely-Guarantee or
Owicki-Gries proofs in Boogie.
Invariants inside a loop with an empty loop-body can be removed by over constraining invariant at the loop entry. This modification will safely under approximate the search space for invariants.
We have adopted this encoding technique to verify Peterson's algorithm.

Table~\ref{table:conc-rec-results} shows the results of running
our Horn-DT tool on these programs.  We note that each of these programs
produce a large number of Horn counter-examples. The Boogie Teacher
first tries to produce positive or negative counter-examples, and
produces Horn counter-examples only if these are not possible.
This indicates that Horn counter-examples are predominantly
needed for these kinds of program specifications.
Our tool successfully learns invariants for all these programs in
reasonable time, with most finishing within 10 seconds. Only one program took more than 50 s to verify.

Verification using Owicki-Gries proof rules requires
adequate invariants at each program point in each
thread.  In comparison, Rely-Guarantee additionally requires two-state
invariants for each thread for the Rely/Guarantee conditions.
These additional invariants makes learning for Rely-Guarantee
proofs more difficult for the Learner.

\bibliographystyle{splncs03}
\bibliography{references}

\clearpage
\appendix

\section{Details of Experimental Evaluation}
\label{app:seq-experiments-details}

The following table lists the results of our experimental evaluation on the sequential programs benchmark suite.

\begin{longtable}{l@{\hskip 1.5em}rrrrrr@{\hskip 1.5em}r}
	\caption{Experimental results of Horn-DT and Seahorn on the sequential programs benchmark suite. ``Rounds'' corresponds to the number of rounds of the learning process. ``Pos'', ``Neg'', and ``Horn'' refers to the number of positive, negative, and Horn examples produced during the learning, respectively. ``To'' indicates a timout after 300\,s, while``FP'' indicates a false positive. All times are given in seconds.}\\
	\endfirsthead
	\caption[]{continued}
	\endhead

	\toprule
	Benchmark & \multicolumn{6}{c}{Horn-DT} & Seahorn \\
	\cmidrule(l{-.5em}r{1.5em}){2-7} \cmidrule(l{-.5em}r){8-8}
	& Rounds & Pos & Neg & Horn & Learner time & Total time & Time \\
	\midrule
	add.bpl & 8 & 1 & 6 & 1 & 0.03 & 0.8 & 4.49 \\
	afnp.bpl & 13 & 1 & 3 & 9 & 0.06 & 0.81 & 0.35 \\
	array2.bpl & 5 & 3 & 2 & 1 & 0.05 & 0.43 & 0.46 \\
	array.bpl & 14 & 4 & 4 & 7 & 0.06 & 0.99 & 0.26 \\
	array\_diff.bpl & 3 & 2 & 2 & 0 & 0.01 & 0.34 & 0.23 \\
	arrayinv1.bpl & 134 & 4 & 32 & 100 & 7.44 & 11.57 & 1.34 \\
	arrayinv2.bpl & 45 & 7 & 12 & 27 & 0.38 & 2.55 & 0.22 \\
	bool\_dec.bpl & 3 & 2 & 1 & 0 & 0.01 & 0.31 & FP \\
	bool\_inc.bpl & 3 & 1 & 1 & 1 & 0.15 & 0.45 & FP \\
	cegar1.bpl & 5 & 3 & 1 & 1 & 0.02 & 0.41 & 0.28 \\
	cegar2.bpl & 24 & 4 & 7 & 14 & 0.28 & 1.67 & 0.30 \\
	cggmp.bpl & 67 & 1 & 13 & 53 & 0.69 & 2.11 & 0.31 \\
	countud.bpl & 25 & 3 & 16 & 8 & 0.11 & 0.57 & 0.27 \\
	dec.bpl & 3 & 1 & 2 & 0 & 0.01 & 0.32 & 0.31 \\
	dillig01.bpl & 8 & 2 & 6 & 1 & 0.03 & 0.37 & 0.22 \\
	dillig03.bpl & 9 & 2 & 4 & 3 & 0.04 & 0.38 & 0.21 \\
	dillig05.bpl & 15 & 3 & 10 & 4 & 0.11 & 0.5 & TO \\
	dillig07.bpl & 12 & 3 & 5 & 5 & 0.05 & 0.4 & 0.28 \\
	dillig15.bpl & 12 & 2 & 3 & 7 & 0.05 & 0.41 & 0.26 \\
	dillig17.bpl & 39 & 5 & 9 & 31 & 0.21 & 0.8 & 0.36 \\
	dillig19.bpl & 12 & 2 & 4 & 7 & 0.05 & 0.43 & 0.31 \\
	dillig24.bpl & 15 & 0 & 0 & 17 & 0.08 & 0.53 & 0.44 \\
	dillig25.bpl & 47 & 1 & 25 & 36 & 0.32 & 1.1 & TO \\
	dillig28.bpl & 25 & 1 & 0 & 34 & 0.12 & 0.61 & TO \\
	dtuc.bpl & 28 & 6 & 15 & 14 & 0.14 & 0.74 & FP \\
	ex14.bpl & 2 & 1 & 1 & 0 & 0.01 & 0.31 & 0.24 \\
	ex14c.bpl & 2 & 1 & 1 & 0 & 0.01 & 0.32 & 0.22 \\
	ex23.bpl & 19 & 2 & 7 & 11 & 0.1 & 0.58 & FP \\
	ex7.bpl & 3 & 2 & 2 & 0 & 0.01 & 0.39 & FP \\
	fig1.bpl & 5 & 2 & 3 & 1 & 0.02 & 0.34 & 0.24 \\
	fig3.bpl & 3 & 2 & 2 & 0 & 0.01 & 0.34 & FP \\
	fig9.bpl & 2 & 1 & 1 & 0 & 0.01 & 0.32 & 0.20 \\
	formula22.bpl & 56 & 1 & 9 & 46 & 0.44 & 1.24 & 0.39 \\
	formula25.bpl & 32 & 1 & 28 & 3 & 0.16 & 0.76 & 0.67 \\
	formula27.bpl & 116 & 1 & 105 & 10 & 0.93 & 2.62 & 0.89 \\
	inc2.bpl & 10 & 4 & 3 & 3 & 0.05 & 0.45 & 0.26 \\
	inc.bpl & 102 & 1 & 1 & 100 & 1.98 & 3.14 & 0.21 \\
	loops.bpl & 16 & 3 & 0 & 16 & 0.08 & 0.53 & 0.35 \\
	\bottomrule
	\pagebreak
	\toprule
	Benchmark & \multicolumn{6}{c}{Horn-DT} & Seahorn \\
	\cmidrule(l{-.5em}r{1.5em}){2-7} \cmidrule(l{-.5em}r){8-8}
	& Rounds & Pos & Neg & Horn & Learner time & Total time & Time \\
	\midrule		
	matrixl1.bpl & 8 & 6 & 7 & 1 & 0.04 & 0.51 & 0.21 \\
	matrixl1c.bpl & 8 & 5 & 11 & 0 & 0.04 & 0.5 & FP \\
	matrixl2.bpl & 24 & 14 & 9 & 6 & 0.11 & 0.87 & 0.21 \\
	matrixl2c.bpl & 36 & 15 & 18 & 6 & 0.2 & 1.31 & FP \\
	nc11.bpl & 9 & 2 & 4 & 4 & 0.04 & 0.39 & 0.26 \\
	nc11c.bpl & 8 & 3 & 3 & 3 & 0.04 & 0.39 & 0.27 \\
	sqrt.bpl & 95 & 6 & 49 & 42 & 1.29 & 3.67 & FP \\
	square.bpl & 18 & 1 & 14 & 3 & 0.08 & 0.79 & FP \\
	sum1.bpl & 22 & 5 & 16 & 4 & 0.1 & 0.63 & 0.30 \\
	sum3.bpl & 6 & 1 & 4 & 1 & 0.03 & 0.38 & 0.26 \\
	sum4.bpl & 41 & 1 & 7 & 33 & 0.21 & 0.98 & 0.28 \\
	sum4c.bpl & 19 & 5 & 11 & 4 & 0.09 & 0.59 & 0.19 \\
	tacas.bpl & 32 & 15 & 9 & 11 & 0.17 & 0.69 & 0.30 \\
	trex1.bpl & 4 & 2 & 3 & 0 & 0.02 & 0.36 & 0.25 \\
	trex3.bpl & 10 & 4 & 7 & 2 & 0.04 & 0.46 & 0.35 \\
	vsend.bpl & 2 & 1 & 1 & 0 & 0.01 & 0.34 & 0.20 \\
	w1.bpl & 4 & 2 & 1 & 1 & 0.02 & 0.34 & 0.29 \\
	w2.bpl & 4 & 1 & 2 & 1 & 0.02 & 0.34 & 0.27 \\
	\bottomrule

\end{longtable}

\end{document}